\documentclass[conference]{IEEEtran}

\usepackage{latexsym,amsmath,amsthm,amssymb}
\def\F{\mathbb{F}}

\def\Z{\mathbb{Z}}

\def\cA{{\cal
A}}\def\cT{{\mathcal T}}
\def\cS{{\mathcal{S}}}

\def\cA{{\mathcal{A}}}
\def\cR{{\mathcal{R}}}

\def\-3{{\sqrt{-3}}}
\def\3{{\sqrt{3}}}

\def\bc{{\bf c}}

\def\b0{{\bf 0}}

\def\Ga{{\alpha}}
\def\Gb{{\beta}}

\newtheorem{thm}{Theorem}[section]

\newtheorem{prop}[thm]{Proposition}
\newtheorem{lem}[thm]{Lemma}

\numberwithin{equation}{section} 
\newtheorem{ex}[thm]{Example}

\date{}
\begin{document}
\title{A Construction of Quantum Codes via A Class of Classical Polynomial Codes}

\author{Lingfei~Jin and Chaoping~Xing
\thanks{L.  Jin and C.  Xing are with Division of
Mathematical Sciences, School of Physical and Mathematical Sciences,
Nanyang Technological University, Singapore 637371, Republic of
Singapore (emails:ljin1@e.ntu.edu.sg; xingcp@ntu.edu.sg).}
\thanks{The work  was partially supported by the Singapore  National Research
Foundation under Research Grant NRF-CRP2-2007-03.}
}

\maketitle

\begin{abstract} There have been various constructions of classical codes from polynomial valuations in literature \cite{ARC04, LNX01,LX04,XF04,XL00}. In this paper, we present a construction of classical codes based on polynomial construction again. One of the features of this construction is that not only the classical codes arisen from the construction have good parameters, but also  quantum codes with reasonably good parameters can be produced from these classical codes. In particular, some new quantum codes are constructed (see Examples \ref{5.5} and \ref{5.6}).
\end{abstract}
\section{Introduction}
One way to produce good quantum codes is to make use of Hermitian self-orthogonal classical codes \cite{Ash Kni}. To get $\ell$-ary quantum codes, one needs Hermitian self-orthogonal classical codes over $\F_{\ell^2}$ with good minimum distance of dual codes. Due to the fact that the Hermitian inner product involves power $\ell$ (see (\ref{eq:4.2})), the parameters of  quantum codes derived from Hermitian self-orthogonal classical codes are usually constrained. For instance, in \cite{Jin Ling Luo Xin} (also see \cite{Gra Bet Roe}), quantum MDS codes  produced by using Hermitian self-orthogonal classical codes have relatively small dimension.

In this paper, we first go to a field of larger size to obtain classical codes over $\F_{\ell^2}$ and then we select Hermitian self-orthogonal  codes from these classical codes over $\F_{\ell^2}$. In this way, we can produce good quantum codes. Our idea to produce classical codes over $\F_{\ell^2}$ from a field of large size has already been studied in the  previous papers \cite{ARC04, LNX01,LX04,XF04,XL00} where polynomial codes were considered. The main idea of this paper is to convert some of these codes into Hermitian self-orthogonal in order to construct quantum codes. It turns out that some new quantum codes can be produced (see Examples \ref{5.5} and \ref{5.6}).

The paper is organized as follows. In Section II, we introduce some
basic notations and results about cyclotomic cosets and corresponding polynomials. In Section III, we show how classical codes can be constructed from these cosets and polynomials. To construct quantum codes, we study dual codes of these classical codes in   Section IV. In the last section, we apply the results in the previous sections to construction of quantum codes.

\section{Cyclotomic cosets and corresponding polynomials}
Let $q$ be a prime power and let $n>1$ be a positive integer with $\gcd(q,n)=1$. Let $m$ be the order of $q$ modulo $n$, i.e, $m$ is the smallest positive integer such that $n$ divides $q^m-1$.

For any $a\in\Z_n$,  we define a $q$-cylotomic coset modulo $n$
\[S_a:=\{a\cdot q^i\bmod{n}:\; i=0,1,2,\dots\}.\]
 It is a well-know fact that all $q$-cyclotomic cosets partition the set $\Z_n$. Let $S_{a_1},S_{a_2},\dots,S_{a_t}$ stand for all distinct $q$-cyclotomic cosets modulo $n$. Then, we have  that $\Z_n=\cup_{i=1}^tS_{a_i}$ and $n=\sum_{i=1}^t|S_{a_i}|$. We denote by $s_a$ the size of the $q$-cyclotomic coset $S_a$.

 The following fact can be easily derived.
 \begin{lem}\label{2.1} For every $a\in\Z_n$, the size $s_a$ of $S_a$ divides $m$ which is the order of $q$ modulo $n$.
 \end{lem}
\begin{proof} It is clear that  $s_a$ is the smallest positive integer such that $a\equiv aq^{s_a}\bmod{n}$, i.e, $s_a$ is the smallest positive integer such that $n/\gcd(n,a)$ divides $q^{s_a}-1$. Since $n/\gcd(n,a)$ also divides $q^m-1$, we have $m\equiv 0\bmod{s_a}$ by applying the long division.
\end{proof}

Now for each $S_a$, we form $s_a$ polynomials in the following way. Let $\Ga_1,\dots,\Ga_{s_a}$ be an $\F_q$-basis of $\F_{q^{s_a}}$ (note that $\F_{q^{s_a}}$ is a subfield of $\F_{q^{m}}$). Consider the polynomials $f_{a,j}(x):=\sum_{i=0}^{s_a-1}\left(\Ga_jx^a\right)^{q^i}$ for $j=1,2,\dots,s_a$.

\begin{lem}\label{2.2} For every $a\in\Z_n$, we have the following facts.
\begin{itemize}
\item[{\rm (i)}] The polynomials $f_{a,j}(x)$ for $j=1,2,\dots,s_a$ are linearly independent over $\F_q$.
\item[{\rm (ii)}] $f_{a,j}(\beta)$ belongs to $\F_q$ for all $\beta\in U_n\cup\{0\}$, where $U_n$ is the subgroup of $n$-th roots of unity in $\F_{q^m}^*$, i.e., $U_n:=\{\Gb\in\F_{q^m}^*:\; \Gb^n=1\}$.
    \end{itemize}
 \end{lem}
\begin{proof} (i) is clear since the coefficients of $x^a$ in  $f_{a,j}(x)$ are $\Ga_j$ and $\Ga_1,\Ga_2,\dots,\Ga_{s_a}$ form  an $\F_q$-basis of $\F_{q^{s_a}}$.

To prove (ii), it is sufficient to prove that $(f_{a,j}(\beta))^q=f_{a,j}(\beta)$ for every $\beta\in U_n\cup\{0\}$. Consider
\begin{eqnarray*} (f_{a,j}(\beta))^q&=&\left(\sum_{i=0}^{s_a-1}\left(\Ga_j\Gb^a\right)^{q^i}\right)^q\\
&=&\sum_{i=0}^{s_a-1}\left(\Ga_j\Gb^a\right)^{q^{i+1}}=\sum_{i=1}^{s_a-1}\left(\Ga_j\Gb^a\right)^{q^{i}}+\Ga_j^{q^{s_a}}\Gb^{aq^{s_a}}\\
&=&\sum_{i=1}^{s_a-1}\left(\Ga_j\Gb^a\right)^{q^{i}}+\Ga_j\Gb^a=f_{a,j}(\beta).
\end{eqnarray*}
This completes the proof.
\end{proof}
\section{Construction of classical codes}
In this section, we give a construction of classical codes basing on the facts from Section 2. For a positive integer $r$ with $1\le r\le n-1$, consider the set of polynomials
\[P_r:=\{f_{a,j}(x):\; 0\le a\le r, \; j=1,2,\dots,s_a\}.\]
Denote the size of $P_K$ by $k_r$. From Lemma \ref{2.2}, it is clear that the polynomial space $V_r$ spanned by $P_r$ over $\F_q$ has dimension $k_r$.

The code $C_r$ is defined by
\begin{equation}\label{eq:3.1}\{(f(\beta))_{\beta\in U_n\cup\{0\}}:\; f\in V_r\}.\end{equation}
\begin{prop}\label{3.1} The code $C_r$ defined in {\rm (\ref{eq:3.1})} is a $q$-ary linear code with parameters $[n+1,k_r,\ge n+1-r]$.
\end{prop}
\begin{proof}
As the degree of every polynomial $f(x)$ in $V_r$ is at most $r\le n-1$, it has at most $r$ roots. Thus, $(f(\beta))_{\beta\in U_n\cup\{0\}}$ has the Hamming weight at least $n+1-r$  as long as $f$ is a nonzero polynomial. Hence, the dimension of $C_r$ is the same as the one of $V_r$, i.e., $\dim(C_r)=k_r$. Moreover, the minimum distance of $C_r$ is at least $n+1-r$.
\end{proof}

\begin{ex}\label{3.2}{\rm Let $q=4$ and $n=51$. Then the order of $4$ modulo $51$ is $m=4$. All $4$-cyclotomic cosets modulo $51$ are
\begin{center} \begin{tabular}{|c|c|c|} \hline
$\{0\}$&$\{1,4,13,16\}$&$\{2,8,26,32\}$\\
$\{3,12,39,48\}$&$\{5,14,20,29\}$&$\{6,24,27,45\}$\\
$\{7,10,28,40\}$&$\{9,15,36,42\}$&$\{11,23,41,44\}$\\
$\{17\}$&$\{18,21,30,33\}$&$\{19,25,43,49\}$\\
$\{22,31,37,46\}$&$\{34\}$&$\{35,38,47,50\}$\\
\hline
\end{tabular}
\end{center}
For instance, for $r=16$, we obtain a $4$-ary $[52,5,\ge 36]$-linear code. This is an optimal code in the sense that for given length and dimension, the minimum distance can not be improved. For  $r=17$, we obtain a $4$-ary $[52,6,\ge 35]$-linear code which is best known based on the online table \cite{Gr12}.
}\end{ex}

\begin{ex}\label{3.3}{\rm Let $q=4$ and $n=63$. Then the order of $4$ modulo $63$ is $m=3$. All $4$-cyclotomic cosets modulo $63$ are
\begin{center} \begin{tabular}{|c|c|c|} \hline
$\{0\}$&$\{1,4,16\}$&$\{2,8,32\}$\\
$\{3,12,48\}$&$\{5,17,20\}$&$\{6,24,33\}$\\
$\{7,28,49\}$&$\{9,18,36\}$&$\{10,34,40\}$\\
$\{11,44,50\}$&$\{13,19,52\}$&$\{14,35,56\}$\\
$\{15,51,60\}$&$\{21\}$&$\{22,25,37\}$\\
$\{23,29,53\}$&$\{26,38,41\}$&$\{27,45,54\}$\\
$\{30,39,57\}$&$\{31,55,61\}$&$\{47,59,62\}$\\
$\{43,46,58\}$&$\{42\}$&\\
\hline
\end{tabular}
\end{center}
For instance, for $r=16$, we get a $4$-ary $[64,4,\ge 48]$-linear code. This is an optimal code in the sense that for given length and dimension, the minimum distance can not be improved. For $r=20$, again we get an optimal $4$-ary $[64,7,\ge 44]$-linear code.  For $r=21$,  an optimal $4$-ary $[64,8,\ge 43]$-linear code can be derived as well.
}\end{ex}

\section{Dual codes}
In this section, we study dual codes for those codes arisen from cyclotomic cosets. From now on, we assume that $q$ is even. Then $n$ is always odd (as $\gcd(n,q)=1$) and hence $n+1$ is even.

Two $q$-cyclotomic cosets $S_a$ and $S_b$ are called {\it dual} if there exists $c\in S_b$ such that $a+c$ is divisible by $n$. For instance, in Example \ref{3.2}, $\{1,4,13,16\}$ and $\{35,38,47,50\}$ are dual to each other. It is clear that the dual of a given cyclotomic coset is unique. Moreover, we have the following facts.
\begin{lem}\label{4.1} Let $S_a$ be the dual of a cyclotomic coset $S_b$. Then we have
\begin{itemize}
\item[{\rm (i)}] $|S_a|=|S_b|$
\item[{\rm (ii)}] For every $x\in S_a$, there exists  $y\in S_b$ such that  $x+y$ is divisible by $n$.
    \end{itemize}
\end{lem}
\begin{proof} We may assume that $a+b$ is divisible by $n$. By definition, $s_b$ is the smallest positive integer  such that $n$ divides $b(q^{s_b}-1)$. Thus, $s_b$ is the smallest positive integer such that $n$ divides $-b(q^{s_b}-1)$. As $-b(q^{s_b}-1)\equiv a(q^{s_b}-1)\bmod{n}$, the desired result of part (i) follows.

Let $x\equiv aq^i\bmod{n}$ for some integer $i$. By definition, there exists $c\in S_b$ such that $a\equiv -c\bmod{n}$. Thus, $x\equiv aq^i\equiv -cq^i\bmod{n}$. Put $y=cq^i\bmod{n}\in S_b$. We obtain the desired result of part (ii).
\end{proof}

Consider a set $\cS$ of cyclotomic cosets such that $\{0\}\in \cS$. Let $\cS^*$ denote the collection of duals of cyclotomic cosets in $\cS$. We denote by $P_{\cS}$ the polynomial set
\[\{f_{a,j}(x):\; S_a\in \cS;\; j=1,2,\dots,s_a\}.\]
Let $V_{\cS}$ be the $\F_q$-space spanned by all polynomials in $P_{\cS}$. Define the $\F_q$-linear code by
\begin{equation}\label{eq:4.1}
C_{\cS}:=\{(f(\beta))_{\beta\in U_n\cup\{0\}}:\; f\in V_{\cS}\}
\end{equation}

Then we have the following result.

\begin{prop}\label{4.2} Let $\cA=\cup_{i=1}^tS_{a_i}$ be the set of all $q$-cyclotomic cosets modulo $n$. Then the Euclidean dual of $C_{\cS}$ is $C_{\cR}$, where $\cR=\{\{0\}\}\cup\left(\cA-\cS^*\right)$.
\end{prop}
\begin{proof} First of all, the dimension of the code $C_{\cS}$ is $\sum_{S\in\cS}|S|$. Thus, the dimension of $C_{\cR}$ is $1+\sum_{S\in \cA}|S|-\sum_{T\in\cS}|T|=n+1-\dim(C_{\cS})$ (note the fact that $\sum_{S\in \cA}|S|=|\Z_n|=n$). To prove our lemma, it is sufficient to show that every codeword in $C_{\cS}$ is orthogonal to all codewords of $C_{\cR}$ under the dot product.

For a polynomial $u(x)$ in $P_{\cA}$, we denote by $\bc_{u}$ the codeword $(u(\beta))_{\beta\in U_n\cup\{0\}}$.  Let $f(x), g(x)$ be  polynomials in $P_{\cS}$ and $P_{\cR}$, respectively. If both $f(x)$ and $g(x)$ are equal to $1$. Then  $\bc_{f}=\bc_{g}$ is the all-one vector ${\bf 1}$. It is clear that in this case $\bc_{f}$ and $\bc_{g}$ are orthogonal  under the dot product. Now assume that at least one of $f(x), g(x)$ is not equal to $1$. Then  for any terms $x^i$ in $f(x)$ and  terms $x^j$ in $g(x)$, we have $i+j\not\equiv 0\bmod{n}$. Thus, the product $f(x)g(x)$ contains only terms $x^k$ with $k\not\equiv 0\bmod{n}$. For such $k$ we have
\[\sum_{\beta\in U_n\cup\{0\}}\Gb^k=\frac {\Ga^{kn}-1}{\Ga^k-1}=0,\]
where $\Ga$ is an $n$-th primitive root of unity in $U_n$. This implies that $\bc_{f}$ and $\bc_{g}$ are orthogonal  under the dot product. The desired result follows.
\end{proof}
\begin{ex}\label{4.3}{\rm Let $q=4$ and $n=51$. Let $\cS=\{\{0\},\{1,4,13,16\}\}$. By Example \ref{3.2}, we know that $\cR=\cA-\{\{35,38,47,50\}\}$.
}\end{ex}

In order to apply our results to quantum codes, we want to discuss the Hermitian dual of $C_{\cS}$ as well. Let us assume that $q$ is equal to $\ell^2$. The Hermitian inner product of the two vectors $(u_1,u_2,\dots,u_{n+1})$ and $(v_1,v_2,\dots,v_{n+1})$ in $\F_{\ell^2}^n$ is defined by
\begin{equation}\label{eq:4.2}\sum_{i=1}^{n+1}u_i^{\ell}v_i.\end{equation}

By abuse of notations, for a set $\cS=\{S_a\}_{a\in I}$ of cyclotomic cosets,  we denote by ${\ell}\cS$ the set $\{S_{a\ell}\}_{a\in I}$ of the cyclotomic cosets .

\begin{prop}\label{4.4} Under the inner product {\rm (\ref{eq:4.2})}, the Hermitian dual of $C_{\cS}$ is $C_{\cT}$, where $\cT=\{\{0\}\}\cup\left(\cA-(\ell\cS)^*\right)$.
\end{prop}
\begin{proof} It is clear that the Hermitian dual of $C_{\cS}$ is the Euclidean dual of $C_{\ell\cS}$. Now the desired result follows from Proposition \ref{4.2}.
\end{proof}

\begin{ex}\label{4.5}{\rm Let $q=4$ and $n=51$. Let $\cS=\{\{0\},\{1,4,13,16\}\}$. By Example \ref{3.2}, we know that $\cT=\cA-\{\{19,25,43,49\}\}$.
}\end{ex}
\section{Application to quantum codes}
In this section, we show how to apply the results from the previous sections to obtain quantum codes.

Instead of giving several complicated  results with detailed formula, we give a general result in this section. Then we use examples to illustrate our result.

\begin{thm}\label{5.1} Let $\cS$ be a set of $q$-cyclotomic cosets modulo $n$ and let $\cT=\{\{0\}\}\cup\left(\cA-(\ell\cS)^*\right)$ such that $(\ell\cS)^*$ contains all cyclotomic cosets $\{S_a:\; n+2-d\le a\le n-1\}$. If $\cS$ is a subset of $\cT$, then there exists an $\ell$-ary quantum code $[[n+1,n+1-2k,\ge d]]$, where $k$ is the $\F_q$-dimension of $C_{\cS}$.
\end{thm}
 \begin{proof}  By Proposition \ref{4.4}, the  Hermtian dual of $C_{\cS}$ is $C_{\cT}$. Under our assumption, $C_{\cS}$ is Hermitian self-orthogonal under the inner product {\rm (\ref{eq:4.2})}. Thus, we obtain  an $\ell$-ary quantum code $[[n+1,n+1-2k]]$ with minimum distance at least the Hamming distance  of $C_{\cT}$ (see \cite{Ash Kni}). As $P_{\cT}$ contains polynomials of degree at most $n+1-d$, the  Hamming distance  of $C_{\cT}$ is at least $d$. This completes the proof.
\end{proof}

\begin{ex}\label{5.2}{\rm Let $q=4$ and $n=21$. Then the order of $4$ modulo $21$ is $m=3$. All $4$-cycloyomic cosets modulo $21$ are
\begin{center} \begin{tabular}{|c|c|c|} \hline
$\{0\}$&$\{1,4,16\}\}$&$\{2,8,11\}$\\
$\{3,6,12\}$&$\{5,17,20\}$&$\{7\}$\\
$\{9,15,18\}$&$\{10,13,19\}$&$\{14\}$\\
\hline
\end{tabular}
\end{center}
Let $\cS=\{\{0\},\{1,4,16\},\{2,8,11\},\{3,6,12\}\}$. Then   $2\cS=\cS$ and
$(2\cS)^*=\{\{0\}, \{5,17,20\},\{10,13,19\},\{9,15,18\}\}.$ Moreover, $\cS$ is contained in $\cT=\{\{0\}\}\cup\left(\cA-(2\cS)^*\right)$. As $S_{17},S_{18},S_{19}$ and $S_{20}$ belong to $(2\cS)^*$, we obtain a binary quantum $[[22,2,6]]$ code which achieves the best-known parameters \cite{Gr12}.
}\end{ex}

\begin{ex}\label{5.3}{\rm Let $q=4$ and $n=51$. Then the order of $4$ modulo $51$ is $m=4$.
Let $\cS=\{\{0\},\{1,4,13,16\},\{2,8,26,32\}, \{6,24,27,45\}\}$. Then   $2\cS=\{\{0\},\{1,4,13,16\},\{2,8,26,32\},\{3,12,39,48\}\}$ and
$(2\cS)^*=\{\{0\},\{35,38,47,50\}, \{19,25,43,49\},\{3,12,39,48\}\}$. Moreover, $\cS$ is contained in $\cT=\{\{0\}\}\cup\left(\cA-(2\cS)^*\right)$. As $S_{47}, S_{48}, S_{49}$ and $S_{50}$ belong to $(2\cS)^*$, we obtain a binary quantum $[[52,26,6]]$ code which meets the best-known one in the online table \cite{Gr12}.

In the similar way, we obtain binary quantum codes with parameters $[[52,24,7]]$ and $[[52,8,10]]$. Both codes meet the parameters of the best-known ones in \cite{Gr12}.
}\end{ex}

\begin{ex}\label{5.4}{\rm Let $q=4$ and $n=63$. Then the order of $4$ modulo $63$ is $m=3$.
\begin{itemize}
\item [(i)] $\cS=\{\{0\},\{1,4,16\},\{2,8,32\}\}$. Then   $2\cS=\cS$ and
$(2\cS)^*=\{\{0\},\{31,55,61\},\{47,59,62\}\}.$ Moreover, $\cS$ is contained in $\cT=\{\{0\}\}\cup\left(\cA-(2\cS)^*\right)$. As $S_{61}$ and $S_{50}$ belong to $(2\cS)^*$, we obtain a binary quantum $[[64,50,4]]$ code which is optimal \cite{Gr12}.
\item [(ii)] $\cS=\{\{0\}\},\{1,4,16\},\{2,8,32\}, \{6,24,33\}\}$. Then   $2\cS=\{\{0\}\},\{1,4,16\},\{2,8,32\},\{3,12,48\}\}$ and
$(2\cS)^*=\{\{0\},\{15,51,60\}\{31,55,61\},\{47,59,62\}\}$. Moreover, $\cS$ is contained in $\cT=\{\{0\}\cup\left(\cA-(2\cS)^*\right)$. As $S_{59}, S_{60}, S_{61}$ and $S_{62}$ belong to $(2\cS)^*$, we obtain a binary quantum $[[64,44,6]]$ code which is optimal again \cite{Gr12}.
\end{itemize}
Analogously, binary quantum codes with parameters $[[64,38,7]]$ and $[[64,32,8]]$ can be derived. Both codes meet the parameters of the best-known ones in \cite{Gr12}.
}\end{ex}

\begin{ex}\label{5.5}{\rm Let $q=16$ and $n=51$. Then the order of $16$ modulo $51$ is $m=2$.
Let $\cS=\{\{0\},\{12,39\},\{8,26\},\{4,13\}\}$. Then   $4\cS=\{\{0\},\{1,16\},\{2,32\},\{3,48\}\}$ and
$(4\cS)^*=\{\{0\},\{3, 48\},\{19, 49\},\{35, 50\}\}$. Moreover, $\cS$ is contained in $\cT=\{\{0\}\}\cup\left(\cA-(2\cS)^*\right)$.  As $S_{50}, S_{49}$ and $S_{48}$ belong to $(4\cS)^*$, we obtain a $4$-ary quantum $[[52,38,5]]$-code.

Likewise, we obtain $4$-ary quantum codes with parameters $[[52,34,6]]$, $[[52,30,7]]$, $[[52,26,8]]$, $[[52,22,9]]$, $[[52,18,10]]$ and $[[52,14,12]]$. The last one meets the parameters of the best-known ones in \cite{Br12} and the rest are new to the online table \cite{Br12}.
}\end{ex}

\begin{ex}\label{5.6}{\rm Let $q=64$ and $n=585$. Then the order of $64$ modulo $585$ is $m=2$.
Let $\cS=\{\{0\},\{8, 512\},\{16,439\}\}$. Then   $8\cS=\{\{0\},\{1,64\},\{2,128\}\}$ and
$(8\cS)^*=\{\{0\},\{457, 583\},\{521, 584\}\}$. Moreover, $\cS$ is contained in $\cT=\{\{0\}\}\cup\left(\cA-(2\cS)^*\right)$.  As $S_{584}$ and $S_{583}$  belong to $(8\cS)^*$, we obtain a $8$-ary quantum $[[586,576,4]]$-code.

In the similar way, we draw $8$-ary quantum codes with parameters $[[586,572,5]]$, $[[586,568,6]]$, $[[586,564,7]]$, $[[586,560,8]]$, $[[586,556,9]]$, $[[586,552,10]]$, $[[586,548,11]]$, $[[586,544,12]]$, $[[586,540,13]]$, $[[586,536,14]]$, $[[586,532,15]]$ and so on. Now, compared with the online table \cite{Br12}, these codes have better parameters. For instance, $8$-quantum codes with the parameters  $[[589,553,4]]$, $[[589,513,6]]$, $[[627,561,5]]$, $[[627,531,6]]$, $[[627,501,7]]$, $[[629,557,6]]$, $[[629,533,7]]$, $[[629,521,8]]$
  are given in \cite{Br12}. We can see that with the same distances our codes have bigger dimensions, but smaller lengths.
}\end{ex}

\end{document}